\documentclass[11pt]{article}
\usepackage[margin=1in]{geometry}
\usepackage{appendix}
\usepackage{babel}
\usepackage{amsthm}
\usepackage{graphicx}
\usepackage{amsmath}
\usepackage{amssymb}
\usepackage{array}
\usepackage[colorlinks=true, citecolor=black, linkcolor=black]{hyperref}
\usepackage{authblk}

\usepackage[linesnumbered,ruled,vlined]{algorithm2e}
\usepackage{algorithmicx}
\usepackage{bm}
\usepackage{tikz}
\usepackage{subcaption}
\usepackage{tabulary}
\usepackage{longtable}
\usepackage{cleveref}

\newtheorem{theorem}{Theorem}[section]

\newtheorem{lemma}[theorem]{Lemma}

\newtheorem{corollary}[theorem]{Corollary}

\newtheorem{definition}{Definition}[section]
\theoremstyle{remark}
\newtheorem{remark}{Remark}[section]
\newtheorem{example}{Example}[section]

\newcommand{\noindentparagraph}[1]{\noindent\textbf{#1.}}

\newcommand\class[1]{$\mathsf{#1}$}

\newcommand{\sentence}{\Psi}
\newcommand{\otc}{\bm{\zeta}} 
\newcommand{\otce}{\zeta} 
\newcommand{\segc}{\bm{\rho}} 
\newcommand{\segce}{\rho} 
\newcommand{\segment}[2]{{#1}\rightsquigarrow{#2}}
\newcommand{\pred}[1]{\mathsf{pred}\left( #1 \right)}
\newcommand{\fomodels}[2]{\mathcal{M}_{#1, #2}}

\newcommand{\weight}{w}
\newcommand{\negweight}{\overline{w}}
\newcommand\wfomc{\mathsf{WFOMC}}
\newcommand\fomc{\mathsf{FOMC}}

\newcommand{\loaxiom}{\mathcal{L}}
\newcommand{\succaxiom}{\mathcal{S}}
\newcommand{\gridaxiom}{\mathcal{G}}
\newcommand{\acyclicityaxiom}{\mathcal{A}}

\newcommand\FOtwo{$\text{FO}^2$}

\newcommand\FOthree{$\text{FO}^3$}
\newcommand\Ctwo{$\text{C}^2$}

\newcommand\tiling[1]{\mathsf{11NMCT}_{(#1)}}

\title{Weighted First Order Model Counting for Two-variable Logic with Axioms on Two Relations\thanks{Authors appear in strict alphabetical order.}}
\author[1]{Qipeng Kuang}\affil{The University of Hong Kong, Hong Kong, China}
\author[2]{V\'aclav K\r{u}la}\affil{Czech Technical University in Prague, Prague, Czech Republic}
\author[2]{Ond\v{r}ej Ku\v{z}elka}
\author[3]{Yuanhong Wang}\affil{Jilin University, Changchun, China}
\author[4]{Yuyi Wang}\affil{CRRC Zhuzhou Institute, Zhuzhou, China}
\date{}

\begin{document}

\maketitle


\begin{abstract}
The Weighted First-Order Model Counting Problem (WFOMC) asks to compute the weighted sum of models of a given first-order logic sentence over a given domain. The boundary between fragments for which WFOMC can be computed in polynomial time relative to the domain size lies between the two-variable fragment (\FOtwo{}) and the three-variable fragment (\FOthree{}). It is known that WFOMC for \FOthree{} is \class{\#P_1}-hard while polynomial-time algorithms exist for computing WFOMC for \FOtwo{} and \Ctwo{}, possibly extended by certain axioms such as the linear order axiom, the acyclicity axiom, and the connectedness axiom. All existing research has concentrated on extending the fragment with axioms on a single distinguished relation, leaving a gap in understanding the complexity boundary of axioms on multiple relations. In this study, we explore the extension of the two-variable fragment by axioms on two relations, presenting both negative and positive results. We show that WFOMC for \FOtwo{} with two linear order relations and \FOtwo{} with two acyclic relations are \class{\#P_1}-hard. Conversely, we provide an algorithm in time polynomial in the domain size for WFOMC of \Ctwo{} with a linear order relation, its successor relation and another successor relation.
\end{abstract}


\section{Introduction}

We consider the Weighted First-Order Model Counting Problem (WFOMC) which asks to compute the weighted sum of models for a given first-order logic sentence over a specified domain, alongside a pair of weighting functions that assign a weight to each model of the sentence. WFOMC serves as a fundamental problem in Statistical Relational Learning \cite{getoor2007introduction} where applications such as Markov Logic Networks \cite{WFOMC-UFO2}, parfactor graphs \cite{poole2003first}, probabilistic logic programs \cite{WFOMC-FO2} and probabilistic databases \cite{VandenBroeck13} can be reduced to WFOMC. Recent work also reveals the potential of WFOMC to contribute to enumerative combinatorics by providing a general framework for encoding counting problems, integer sequences, and computing graph polynomials on specific graphs \cite{fluffy,WFOMC-polys}.

However, it is unlikely to find an algorithm computing WFOMC in time polynomial in the size of the sentence due to the fact that WFOMC encodes the well-known \class{\#P}-complete problem \#SAT even if the symbols of relations are fixed \cite{WFOMC-FO3}. Consequently, most study concentrates on the complexity in terms of the domain size, which is an analog to the concept of data complexity in database theory \cite{datacomplexity}. The fragments enjoying polynomial time complexity in the domain size are called \emph{domain-liftable} \cite{WFOMC-UFO2}.

It is therefore interesting to investigate the boundary of domain-liftability for WFOMC within the fragments of first-order logic. Previous results have shown that the two-variable fragment (\FOtwo{}) is domain-liftable \cite{WFOMC-UFO2,WFOMC-FO2}, while the three-variable fragment (\FOthree) is not domain-liftable unless \class{\#P_1 \subseteq FP} \cite{WFOMC-FO3}. Here, \class{\#P_1} defined by \cite{counting-complexity} is the class of counting problems whose input is in unary. The work in \cite{WFOMC-C2} extends the domain-liftability of \FOtwo{} by incorporating cardinality constraints and counting quantifiers, known as \Ctwo{}.

To further advance domain-liftability, recent studies have focused on augmenting \Ctwo{} with \emph{additional axioms} which may not necessarily be finitely expressible in first-order logic. Starting from the result \cite{DBLP:conf/lics/KuusistoL18} showing that \FOtwo{} remains domain-liftable with the addition of a functionality axiom, subsequent results have identified various tractable axioms such as the tree axiom \cite{WFOMC-tree-axioms}, the linear order axiom \cite{WFOMC-linearorder-axiom}, the connectedness axiom \cite{WFOMC-axioms} and the acyclicity axiom \cite{WFOMC-axioms}. The work in \cite{WFOMC-polys} proposes a general approach to prove domain-liftability of axioms expressible by graph polynomials (e.g., the bipartite axiom and the strong connectedness axiom) and even their combinations (e.g., the combination of the acyclicity and the weak connectedness axioms on a single relation).

All the existing results, however, only allow axioms on a single distinguished relation in the sentence. The techniques involved in the above work such as the Matrix-Tree Theorem for the tree axiom \cite{WFOMC-tree-axioms}, the recursion formulae for the connectedness axiom and the acyclicity axiom \cite{WFOMC-axioms}, and polynomials for WFOMC for axioms in \cite{WFOMC-polys} cannot be generalized to axioms on multiple relations trivially. The extended constraints on the first-order sentence involving \emph{multiple relations} such as ``two binary relations are interpreted as linear orders'' have been barely scratched.

In this work, we augment \FOtwo{} and \Ctwo{} with axioms on two distinguished binary relations and examine the boundary of domain-liftability for WFOMC. The axioms we consider include the linear order axiom, the acyclicity axiom, and the successor axiom. The linear order axiom, expressible in the first-order fragments with at least three variables, enables an ordering of domain elements using two variables. The acyclicity axiom requires the relation to represent a directed acyclic graph, therefore it can encode the linear order axiom since a linear order relation can be viewed as an acyclic tournament with self-loops. The successor axiom is a weakening of the linear order axiom by allowing access only to the successor relation of the linear order.

A closely related line of research is the exploration of decidability in the finite satisfiability problem for first-order logics, whose frontier also falls between \FOtwo{} and \FOthree{} due to the decidability of \FOtwo{} \cite{fs-FO2} and the undecidability of \FOthree{} \cite{fs-FO3}. It has been proved that \FOtwo{} with 8 linear order axioms \cite{fs-FO2+8order}, \FOtwo{} with 3 linear order axioms \cite{fs-FO2+lo}, \FOtwo{} with 2 linear order axioms and their successor axioms \cite{fs-FO2+2successor} and \Ctwo{} with 2 linear order axioms \cite{fs-C2+lo} are all undecidable. The proofs of undecidability generally follow this approach: First design a sentence in the fragment so that each model is grid-like (i.e., a grid can be homomorphically embedded in the model), and then use the grid to encode an undecidable tiling problem. In this sense, deciding whether the fragment has a finite model is equivalent to solving the tiling problem.

At first glance, it may seem that the undecidability of a fragment immediately implies impossibility of its domain-liftability for WFOMC, as we can follow the same route that we identify a counting variant of the tiling problem which is hard to compute and then encode the tiling problem using the same techniques applied in undecidability. However, such techniques of encoding a grid using the axioms are not applicable in model counting. In fact, the grid-like models do \emph{not} ensure that the elements form a real grid. Although the existence of a finite model is equivalent to the existence of a valid tiling, which is sufficient for decision problems, there is no clear quantitative relation between the number of grids and the number of models of the sentence, thus computing FOMC or WFOMC for such sentences does not aid in computing the number of tilings. Therefore, a more refined encoding is necessary for model counting.

\subsection{Our Contributions}

We consider WFOMC for \FOtwo{} with linear order axioms and successor axioms on two distinguished binary relations, providing both negative and positive results.

For the negative results, we prove that WFOMC for \FOtwo{} with two linear order relations is \class{\#P_1}-hard. We obtain the hardness in three steps. First, we construct a variant of the tiling problem, 1-1-N-M counting tiling problem, and show that it is \class{\#P_1}-hard to compute. Next, we introduce an intermediate axiom, the grid axiom, and show that WFOMC for \FOtwo{} with a grid axiom is also \class{\#P_1}-hard. Finally, we reduce WFOMC for \FOtwo{} with a grid axiom to WFOMC for \FOtwo{} with two linear order relations and obtain its \class{\#P_1}-hardness. As a linear order axiom can be encoded by an acyclicity axiom, we also conclude that WFOMC for \FOtwo{} with two acyclic relations is \class{\#P_1}-hard.

For the positive results, we give an algorithm computing WFOMC for \FOtwo{} with a linear order relation and a successor relation (of another linear order) in time polynomial in the domain size. There are two implications of this result. First, combining the reduction from WFOMC for \Ctwo{} with cardinality constraints to WFOMC for \FOtwo{} given by \cite{WFOMC-C2}, it is implied that \Ctwo{} with a linear order relation and another successor relation (possibly with cardinality constraints) is domain-liftable for WFOMC as well. Furthermore, the technique of obtaining the successor relation of a linear order proposed by \cite{WFOMC-linearorder-axiom} implies the domain-liftability of \Ctwo{} with a linear order relation, its successor relation and another successor relation (possibly with cardinality constraints).

Our results indicate the boundary of domain-liftability in WFOMC for first-order sentences with linear order axioms in several aspects. From the perspective of the quantity of linear order axioms, our negative result is a complement of the domain-liftability of WFOMC for \FOtwo{} with one linear order relation shown by \cite{WFOMC-linearorder-axiom}. From the perspective of the power of linear order axioms, our results show that the intractability of two linear order axioms arises from the extra information beyond the successor relation of any of the linear orders.

\subsection{Related Work}

The first step in our negative result involves the trick of encoding a Turing machine using a tiling problem. For details of the unbounded tiling problem to which the halting problem can be reduced, one can refer to \cite{tiling-NPC,wang-tile}. Additionally, \cite{tiling-NPC,tiling-bounded} discuss the bounded tiling problem encoding a nondeterministic Turing machine.

The second and third parts of our negative result involve the trick of encoding a tiling problem using a two-variable first-order sentence with multiple axioms that restrict the structures to be grid-like. This trick is commonly used in finite satisfiability \cite{fs-FO2+8order,fs-FO2+lo,fs-FO2+2order,fs-FO2+3eq,fs-FO2+2eq}. However, the trick for decision problems does not guarantee any quantitative relation between the number of valid tilings and the number of models of the sentence.

Our positive result is based on the dynamic programming algorithm computing WFOMC for \FOtwo{} with a single linear order axiom \cite{WFOMC-linearorder-axiom}. Another work \cite{WFOMC-polys} proves the domain-liftability of this fragment as well by defining polynomials for WFOMC, though, their polynomials cannot be easily generalized to accommodate axioms on two relations while being computed in time polynomial in the domain size. 

\section{Preliminaries}



Throughout this paper, we denote the set of natural numbers from 1 to $n$ as $[n]$. All vectors used in this paper are non-negative integer vectors. For a vector $\bm{v}$ of length $L$, $|\bm{v}|$ denotes the sum of its elements, i.e., $\sum_{i=1}^L v_i$. By $\bm{v} \ge 0$ we require that every coordinate of $\bm v$ is non-negative.

\subsection{First-Order Logic with Axioms}\label{sec:fol}

We consider the function-free and finite domain fragment of first-order logic. An \emph{atom} of arity $k$ takes the form $P(x_1, \cdots, x_k)$ where $P/k$ is from a vocabulary of \emph{relations} (also called \emph{predicates}), and $x_1, \cdots, x_k$ are logical variables from a vocabulary of variables or constants from a vocabulary of constants.\footnote{We restrict $k \ge 1$ for simplicity but $k \ge 0$ is tractable for this work as well.} A \emph{literal} is an atom or its negation. A \emph{formula} is a literal or formed by connecting one or more formulas together using negation, conjunction, or disjunction. A formula may be optionally surrounded by one or more quantifiers of the form $\forall x$ or $\exists x$, where $x$ is a logical variable. A logical variable in a formula is \emph{free} if it is not bound by any quantifier. A formula with no free variables is called a \emph{sentence}.

A first-order logic formula in which all variables are substituted by constants in the domain is called \emph{ground}. A \emph{possible world} $\omega$ interprets each relation in a sentence over a finite domain, represented by a set of ground literals. The possible world $\omega$ is a \emph{model} of a sentence $\sentence$ if $\omega \models \sentence$. Here, $\omega \models \alpha$ means that the formula $\alpha$ is true in $\omega$. We denote the set of all models of a sentence $\sentence$ over the domain $\{1,2,\dots,n\}$ by $\fomodels{\sentence}{n}$.

In this paper, we are specially interested in the following fragments of first-order sentences. A sentence with at most two logical variables is called an \FOtwo{} sentence. An \FOtwo{} sentence with \emph{counting quantifiers} $\exists_{=k}$, $\exists_{\le k}$ and $\exists_{\ge k}$ is called a \Ctwo{} sentence where $\exists_{=k}$ restricts that the number of assignments of the quantified variable satisfying the subsequent formula is exactly $k$ and the other two counting quantifiers are defined similarly. An \FOtwo{} or \Ctwo{} sentence can be possibly augmented with \emph{cardinality constraints}, which are expressions of the form of $|P| \bowtie k$ where $P$ is a relation and $\bowtie$ is a comparison operator $\{<, \le, =, \ge, >\}$. These constraints are imposed on the number of distinct positive ground literals of $P$ in a model.

The interpretation of a binary relation $R$ can be regarded as a directed graph $G(R)$ where the domain is the vertex set and the true ground literals of $R$ are the edges. An \emph{axiom} is a special constraint on a binary relation $R$ that $G(R)$ should be a certain combinatorial structure. In this work, we consider the following axioms:

\begin{definition}[Linear Order Axiom \cite{WFOMC-linearorder-axiom}]
The \emph{linear order axiom}, denoted by $\loaxiom(R)$, requires that $R$ should be a linear order of elements, i.e., $R$ is a reflexive, anti-symmetric, transitive and total relation. In other words, $G(R)$ is an acyclic tournament with a self-loop for each vertex.
\end{definition}

\begin{definition}[Successor Axiom]\label{def:succ}
The \emph{successor axiom}, denoted by $\succaxiom(R)$, requires that $R$ should be the successor relation of a linear order. That is, let $L$ be a linear order relation, then $R$ satisfies
\begin{equation*}
  \forall x \forall y \Bigl(R(x,y) \leftrightarrow \bigl(L(x,y) \land \lnot \exists z (L(x,z) \land L(z,y))\bigl)\Bigl).
\end{equation*}

In other words, $G(R)$ is a directed path.
\end{definition}

\begin{definition}[Acyclicity Axiom \cite{WFOMC-polys}]
The \emph{acyclicity axiom}, denoted by $\acyclicityaxiom(R)$, requires that $G(R)$ is acyclic.
\end{definition}

We abbreviate ``a linear order axiom on a distinguished binary relation'' as ``a linear order relation'' for convenience, and deal with the successor axiom and the acyclicity axiom similarly.

\subsection{Unweighted and Weighted First Order Model Counting}\label{sec:wfomc}

The \textit{first-order model counting problem} (FOMC) asks to compute the number of models of a first-order sentence over the domain of a given size.

\begin{definition}[First Order Model Counting]
The FOMC of a first-order sentence $\sentence$ over a finite domain of size $n$ is
\begin{equation*}
  \fomc(\sentence, n) = |\fomodels{\sentence}{n}|.
\end{equation*}
\end{definition}

The \textit{weighted first-order model counting problem} (WFOMC) additionally expects a pair of weighting functions $(\weight, \negweight)$ that both map all relations in $\sentence$ to real weights.
Given a set $L$ of literals whose relations are in $\sentence$, the weight of $L$ is defined as
\begin{equation*}
W(L, \weight, \negweight):= \prod_{l \in L_T}\weight(\pred{l}) \cdot \prod_{l \in L_F}\negweight(\pred{l}),
\end{equation*}
where $L_T$ (resp. $L_F$) denotes the set of positive (resp. negative) literals in $L$, and $\pred{l}$ maps a literal $l$ to its corresponding relation name. We omit the symbols $\weight, \negweight$ and write $W(L)$ in short when the weighting functions are clear in the text.

\begin{example}\label{exmp:model-weight}
Consider the sentence $\sentence = \forall x \forall y (S(x) \to R(x,y))$ and the weighting functions $\weight(S) = 3,\ \weight(R) = 2,\ \negweight(S) = \negweight(R) = 1$.
The weight of the literal set
$$
L = \{S(1), \lnot S(2), R(1,1), R(1,2), R(2,1), \lnot R(2,2)\}
$$
is
\begin{equation*}
  \weight(S) \cdot \left(w(R)\right)^3 \cdot \negweight(S) \cdot \negweight(R) = 24.
\end{equation*}
\end{example}

\begin{definition}[Weighted First Order Model Counting]
    The WFOMC of a first-order sentence $\sentence$ over a finite domain of size $n$ under weighting functions $\weight, \negweight$ is
    \begin{equation*}
        \wfomc(\sentence, n, \weight, \negweight) := \sum_{\mu \in \fomodels{\sentence}{n}} W(\mu, \weight, \negweight).
    \end{equation*}
\end{definition}

As the weighting functions are defined in terms of relations, all positive ground literals of the same relation get the same weights, and so do all negative ground literals of the same relation.
Therefore, the WFOMC we consider is also referred to as \emph{symmetric} WFOMC~\cite{WFOMC-FO3}.

\begin{example}
Consider the sentence $\sentence$ and the weighting functions in \Cref{exmp:model-weight}. Then
\begin{equation*}
\wfomc(\sentence, n, \weight, \negweight) = (3 \cdot 2^n + 3^n)^n.
\end{equation*}

In fact, for each domain element $i \in [n]$, either $S(i)$ is true and $R(i,j)$ is true for all $j \in [n]$ which contributes weight $3 \cdot 2^n$, or $S(i)$ is false and $R(i,j)$ is not limited which contributes weight $(1+2)^n$. Multiplying the contributed weight of each element, we get the above value.
\end{example}

It is worth noting that FOMC is a special case of WFOMC where the weighting functions simply map all relations to $1$.

\subsection{Complexity}

In this work, we consider the complexity for WFOMC with respect to the domain size. That is to say, when measuring the complexity for WFOMC, the problem can be regarded as fixing the sentence and the weighing functions, only inputting $n$ in unary. A sentence, or class of sentences, is said to be \emph{domain-liftable}~\cite{WFOMC-UFO2} for WFOMC if for any pair of weighting functions its WFOMC can be computed in time polynomial in the domain size.

A complexity class relevant for such counting problems with input in unary is the class \class{\#P_1} defined in \cite{counting-complexity}. A function $f: \{1\}^* \to \mathbb N$ is in \class{\#P_1} iff there is a polynomial time nondeterministic Turing machine whose input alphabet is unary such that $f(x)$ equals the number of accepting paths of the Turing machine on the input $x$. A problem is \class{\#P_1}-hard if all \class{\#P_1} functions can be computed in polynomial time with the oracle for this problem. Specifically, WFOMC of a sentence or a class of sentences is \class{\#P_1}-hard if all \class{\#P_1} functions can be computed in polynomial time with the oracle for WFOMC of the sentence(s) with any pair of weighting functions. The following corollary serves as evidence showing that \class{\#P_1} contains functions that are hard to compute.

\begin{corollary}{(Corollary from \cite[Theorem 1]{tally})}
If \class{\#P_1 \subseteq FP}, then \class{E=NE}.
\end{corollary}

The previous work \cite{WFOMC-FO3} showed that there is a universal Turing machine that can simulate all Turing machines representing \class{\#P_1} problems.

\begin{lemma}{(Universal Turing Machine for \class{\#P_1} \cite{WFOMC-FO3})}\label{lemma:utm}
There is a multi-tape linear time nondeterministic Turing machine $M$ such that the input alphabet of $M$ is unary and computing its number of accepting paths on a given input is \class{\#P_1}-hard.
\end{lemma}

\begin{remark}\label{remark:wfomc-membership}
FOMC of a fixed first-order logic sentence is a \class{\#P_1} problem. In fact, a nondeterministic Turing machine can guess a model of the sentence and then verify it in polynomial time of the domain size. Therefore, the number of accepting paths of this Turing machine equals the number of models of the sentence. However, WFOMC of a fixed first-order logic sentence and fixed weighting functions is not necessarily in \class{\#P_1} since the number of accepting paths of any Turing machine must be a natural number but the weights in WFOMC can be negative or even non-integers.
\end{remark}

Let $\sentence, \sentence'$ be two first-order logic sentences possibly with cardinality constraints and axioms, and $n$ be the domain size. We say that WFOMC of $(\sentence, n)$ can be \emph{reduced} to WFOMC of $(\sentence', n)$ if for any pair of weighting functions $\weight, \negweight$, $\wfomc(\sentence, n, \weight, \negweight)$ can be computed in time polynomial in $n$ with the oracle for $\wfomc(\sentence', n, \weight', \negweight')$ for any pair of weighting functions $\weight',\negweight'$.

Previous work \cite{WFOMC-C2} has shown that WFOMC of \Ctwo{} with cardinality constraints can be computed by WFOMC of \FOtwo{}, regardless of the axioms involved in the sentence.

\begin{lemma}{(Eliminating Counting Quantifiers and Cardinality Constraints \cite[Proposition 5 and Theorem 4]{WFOMC-C2})}\label{lemma:c2+cc}
For any \Ctwo{} sentence $\sentence$ (possibly with cardinality constraints) and the conjunction of axioms $A$, there is an \FOtwo{} sentence $\sentence'$ without counting quantifiers or cardinality constraints such that for any domain size $n$, WFOMC of $(\sentence \land A, n)$ can be reduced to WFOMC of $(\sentence' \land A, n)$.
\end{lemma} 

\section{Hardness of Two Linear Order Axioms}
\label{sec:2lo}

In this section, we show \class{\#P_1}-hardness of WFOMC for \FOtwo{} with two linear order relations. Similarly to the proofs of undecidability results \cite{fs-FO2+8order,fs-FO2+lo}, we prove the hardness by reducing WFOMC of a certain sentence to a hard tiling problem. We first define a tiling problem whose number of valid tilings is \class{\#P_1}-hard to compute. Then we introduce an intermediate axiom, the grid axiom, which requires that domain elements form a grid, enabling binary relations $H$ and $V$ to represent the horizontal and vertical successor relations respectively. We show that WFOMC of \FOtwo{} sentences with a grid axiom is \class{\#P_1}-hard by encoding the \class{\#P_1}-hard tiling problem using the grid. Subsequently, we show that the grid axiom can be implemented by two linear order relations, hence the hardness follows.

\subsection{The Tiling Problem}

The hardness comes from a variant of the tiling problem. A \emph{tiling system} is a tuple $(\mathcal T, R_H, R_V)$ where $\mathcal T$ is a set of tiles and $R_H, R_V \subseteq \mathcal T \times \mathcal T$ are the horizontal and vertical constraints. A \emph{tiling} of $(\mathcal T, R_H, R_V)$ on a grid of $n$ rows and $m$ columns is a mapping $T: [n] \times [m] \to \mathcal T$ such that
\begin{itemize}
  \item $(T(i,j), T(i,j+1)) \in R_H$ for each $1 \le i \le n$ and $1 \le j < m$, and
  \item $(T(i,j), T(i+1,j)) \in R_V$ for each $1 \le i < n$ and $1 \le j \le m$.
\end{itemize}

Intuitively, the problem is to assign a set of tiles (e.g., we can think of a tile as a square with colors on its four edges, which is the classical Wang tiles \cite{wang-tile}) to a grid under the constraint that only specified pairs of tiles can be horizontally or vertically adjacent.

The \emph{bounded tiling problem} inputs the tiling system $(\mathcal T, R_H, R_V)$, the size of the grid $n\times n$ and the first row of the tiling $t_1, \cdots, t_n$, and asks if there is a tiling on the grid consistent to the given first row. It is proved in \cite{tiling-NPC} that the problem is \class{NP}-complete as it encodes the accepting paths of a polynomial time nondeterministic Turing machine. Moreover, the properties of the encoding can be summarized in the following lemma.

\begin{lemma}{(Summarization of the Proof of \cite[Theorem 7.2.1]{tiling-NPC})}\label{lemma:bounded-tiling}
For any one-tape nondeterministic Turing machine, there is a tiling system $(\mathcal T, R_H, R_V)$ such that any input $x$ of the Turing machine can be transformed to a tile sequence $t_1, \cdots, t_{p+2} \in \mathcal T$ where $p$ is the time bound for running $x$, and each accepting path of the Turing machine on input $x$ uniquely corresponds to a tiling with the first row $t_1, \cdots, t_p$ and vice versa. In particular, among $t_3, \cdots, t_{|x|+2}$ each tile represents a bit of $x$, and $t_{|x|+3}, \cdots, t_p$ are the same tile representing the default empty symbol on the tape.
\end{lemma}

We now define a problem of counting the tilings where the first row of the grid is restricted to a certain format.

\begin{definition}[1-1-N-M Counting Tiling Problem]
Given a tiling system $(\mathcal{T}, R_H, R_V)$ and four distinguished tiles $t_1, t_2, t_3, t_4 \in \mathcal T$, the \emph{1-1-N-M counting tiling problem} $\tiling{\mathcal T, R_H, R_V, t_1, t_2, t_3, t_4}$ inputs two positive integers $n,m$ in unary and asks to compute the number of tilings of $(\mathcal{T}, R_H, R_V)$ on a grid of $n+m+2$ rows and $n+m+2$ columns where the first row consists of one tile of $t_1$, one tile of $t_2$, $n$ tiles of $t_3$ and $m$ tiles of $t_4$ in order.
\end{definition}

Since the 1-1-N-M counting tiling problem inputs $n$ and $m$ in unary, we measure its time complexity in terms of $n+m$.

\begin{lemma}\label{lemma:tiling}
There is a \class{\#P_1}-hard 1-1-N-M counting tiling problem $\tiling{\mathcal T, R_H, R_V, t_1, t_2, t_3, t_4}$.
\end{lemma}

The lemma is proved mainly by encoding the Turing machine in \Cref{lemma:utm} in the same way as in \Cref{lemma:bounded-tiling} so that each tiling of certain $n,m$ corresponds to an accepting path of the Turing machine on the input $n$ and time and space bound $n+m+2$.

\begin{proof}
Let $M$ be the Turing machine in \Cref{lemma:utm}. $M$ can be transformed to a one-tape quadratic time Turing machine $M'$ which preserves all other properties of $M$ (nondeterminism, unary input alphabet and \class{\#P_1}-hardness of computing the number of accepting paths) by some well-known tricks (e.g., see \cite[Claim 1.6]{computational-complexity}). Let $c$ be the constant such that $M'$ always terminates within $cn^2$ steps for any input of size $n$ and the length of the tape is bounded by $cn^2$.

The same encoding in \Cref{lemma:bounded-tiling} can be applied to $M'$ since $M'$ is also a one-tape nondeterministic Turing machine. Let $(\mathcal T, R_H, R_V)$ be the tiling system to encode $M'$. Moreover, as the input of $M'$ is in unary, we use the same encoding to obtain the first row $t_1, \cdots, t_{cn^2+2}$ so that $t_3 = \cdots = t_{n+2}$ is the tile representing the input symbol, and $t_{n+3} = \cdots = t_{cn^2+2}$ is the tile representing the default empty symbol. By the same argument in \cite[Theorem 7.2.1]{tiling-NPC}, each tiling on the $(cn^2+2)\times (cn^2+2)$ grid with the specified first row corresponds to an accepting path of $M'$ on the input size $n$, hence $\tiling{\mathcal T, R_H, R_V, t_1, t_2, t_3, t_{n+3}}$ on input $n$ and $m = cn^2-n$ equals the number of accepting paths of $M'$ on the input size $n$.
\end{proof}

\subsection{The Grid Axiom}

The grid axiom requires that the elements of the domain are arranged into a grid and the horizontal and vertical successor relations can be accessed by the binary relations $H$ and $V$ respectively.

\begin{definition}{(The Grid Axiom)}\label{def:grid}
The \emph{grid axiom} over a domain of size $n^2$, denoted by $\gridaxiom(H,V)$, involves two binary relations $H$ and $V$ such that there is a way to arrange the $n^2$ elements into a grid of $n$ rows and $n$ columns, and
\begin{itemize}
  \item $H(x,y)$ is true if and only if $x$ is at the $i$-th row and $j$-th column and $y$ is at the $i$-th row and $j+1$-th column, for some $1 \le i \le n$ and $1 \le j < n$, and
  \item $V(x,y)$ is true if and only if $x$ is at the $i$-th row and $j$-th column and $y$ is at the $i+1$-th row and $j$-th column, for some $1 \le i < n$ and $1 \le j \le n$.
\end{itemize}
\end{definition}

With the grid axiom, we are able to encode the 1-1-N-M counting tiling problem in \Cref{lemma:tiling} such that each valid tiling corresponds to the same number of models of the sentence, and thus computing the number of models of the sentence is as hard as computing the 1-1-N-M counting tiling problem.

\begin{lemma}\label{lemma:grid-tm}
There is an \FOtwo{} sentence with a grid axiom whose WFOMC is \class{\#P_1}-hard.
\end{lemma}

\begin{proof}

Let $\tiling{\mathcal T, R_H, R_V, t_1, t_2, t_3, t_4}$ be the \class{\#P_1}-hard 1-1-N-M counting tiling problem in \Cref{lemma:tiling}. We now encode $\tiling{\mathcal T, R_H, R_V, t_1, t_2, t_3, t_4}$ for the input $n,m$ by an \FOtwo{} sentence with a grid axiom $\gridaxiom(H,V)$ over a domain of size $(n+m+2)^2$. Elements are expected to be arranged into $n+m+2$ rows and $n+m+2$ columns, where the element at the $i$-th row and the $j$-th column represents the cell at the $i$-th row and the $j$-th column of the tiling grid. Let $\mathcal T$ be $\{t_1, \cdots, t_k\}$. We introduce fresh unary relations $T_1(x), \cdots, T_k(x)$ where $T_i(x)$ being true indicates that the cell $x$ is mapped to the tile $t_i$.

Formally, we define the following sentences.

\begin{itemize}
  \item The top and left borders of the grid are identified by unary relations $Top$ and $Left$ according to $H$ and $V$.
  \begin{equation}\label{eq:tm1}
    \begin{aligned}
      & \forall x (Top(x) \leftrightarrow \forall y \lnot V(y,x)) \\
      \land & \forall x (Left(x) \leftrightarrow \forall y \lnot H(y,x)).
    \end{aligned}
  \end{equation}
  \item Each cell is mapped to a unique tile.
  \begin{equation}\label{eq:tm2}
    \forall x \bigvee_{i=1}^k \left( T_i(x) \land \bigwedge_{j \neq i} \lnot T_j(x) \right).
  \end{equation}
  \item The tiling satisfies the constraints $R_H, R_V$.
  \begin{equation}\label{eq:tm3}
    \begin{aligned}
        & \forall x \forall y \bigl(H(x,y) \to \bigvee_{(t_i, t_j \in R_H)} (T_i(x) \land T_j(y))\bigl) \\
        \land & \forall x \forall y \bigl(V(x,y) \to \bigvee_{(t_i, t_j \in R_V)} (T_i(x) \land T_j(y))\bigl).
    \end{aligned}
  \end{equation}
  \item The first row consists of one tile of $t_1$, one tile of $t_2$, $n$ tiles of $t_3$ and $m$ tiles of $t_4$ in order.
  \begin{equation}\label{eq:tm4}
    \begin{aligned}
        & \forall x \bigl(R_1T_1(x) \leftrightarrow (Top(x) \land Left(x))\bigl) \\
        \land & \forall x \bigl(R_1T_2(x) \leftrightarrow \exists y (H(y,x) \land R_1T_1(y))\bigl) \\
        \land & \forall x \Bigl(R_1T_3(x) \leftrightarrow \exists y \bigl(H(y,x) \land (R_1T_2(y) \lor R_1T_3(y))\bigl)\Bigl) \\
        \land & \forall x \Bigl(R_1T_4(x) \leftrightarrow \exists y \bigl(H(y,x) \land (R_1T_3(y) \lor R_1T_4(y))\bigl)\Bigl) \\
        \land & |R_1T_3| = n \\
        \land & |R_1T_4| = m \\
        \land & \bigwedge_{i=1}^4 \forall x (R_1T_i(x) \to T_i(x)),
    \end{aligned}
  \end{equation}
  where $R_1T_i$ ($i = 1,2,3,4$) are auxiliary fresh unary relations.
\end{itemize}


Let $\sentence$ be the conjunction of \Cref{eq:tm1,eq:tm2,eq:tm3,eq:tm4}. Each valid tiling of $\tiling{\mathcal T, R_H, R_V, t_1, t_2, t_3, t_4}$ is encoded by $((n+m+2)^2)!$ models of $\sentence \land \gridaxiom(H,V)$ over a domain of size $(n+m+2)^2$ due to the $((n+m+2)^2)!$ ways to organize a grid by $(n+m+2)^2$ elements. Consequently, any algorithm computing FOMC of $\sentence \land \gridaxiom(H,V)$ in time polynomial in the domain size yields an algorithm to compute the number of valid tilings of $\tiling{\mathcal T, R_H, R_V, t_1, t_2, t_3, t_4}$ in time polynomial in $n+m$, thus the FOMC is \class{\#P_1}-hard.

Note that $\sentence$ is an \FOtwo{} sentence with a cardinality constraint $|Row1T4| = m$. By \Cref{lemma:c2+cc}, there is another \FOtwo{} sentence $\sentence'$ without cardinality constraints such that WFOMC of $\sentence'$ with a grid axiom is \class{\#P_1}-hard as well.
\end{proof}

\subsection{Implementing the Grid Axiom by Two Linear Order Axioms}

Next, we show that the grid axiom over a domain of $n^2$ elements can be further encoded by two linear order relations in WFOMC.

Let us discuss in detail why similar encodings of a grid in the literature such as using eight linear orders \cite{fs-FO2+8order} or three linear orders \cite{fs-FO2+lo} are not applicable in model counting. These encodings aim to ensure that every model of the constructed sentence satisfies a property called ``grid-like'' defined in \cite{fs-FO2+8order}. A model is grid-like if a grid can be homomorphically embedded into it, i.e., there is a mapping $G$ from the grid cells to domain elements along with two binary relations $H,V$ such that the horizontal (resp. vertical) successor relation of any pair of cells $x,y$ implies $H(G(x),G(y))$ (resp. $V(G(x),G(y))$. Using the relations $H$ and $V$ one can encode the tiling problem using a similar technique as in the previous subsection.

However, $G$ is neither required to be injective nor surjective, and the encoding makes no restriction on unmapped elements or elements that are not horizontally or vertically adjacent. This results in non-unified numbers of ways to extract a grid from a model and allows for the same grid to be extracted from different models. For example, one can split the elements into four groups, each encoding a grid. Additionally, based on a model encoding a valid tiling, one can also add more true ground literals of $H$ and $V$ as much as possible. Therefore, computing FOMC or WFOMC for such sentences does not make sense when determining the number of valid tilings.

We, therefore, give a new encoding of a true grid defined by \Cref{def:grid} using two linear order relations such that each grid exactly corresponds to $(n^2)!$ models and thus the number of accepting paths of the Turing machine encoded by the grid can be counted precisely. \Cref{fig:L1L2} shows the target shape of our grid with two linear orders.

Our encoding benefits from the tractability of emulating arbitrary counting quantifiers and cardinality constraints in polynomial time with oracles for WFOMC as stated in \Cref{lemma:c2+cc}, which is not accessible in the finite satisfiability problem when we only have access to the SAT oracle for \FOtwo{}.

\begin{lemma}\label{lemma:grid-2lo}
Let $\sentence$ be an \FOtwo{} sentence with distinguished binary relations $H$ and $V$, and $n \ge 2$ be an integer.
There is an \FOtwo{} sentence $\sentence'$ with distinguished binary relations $L_1$ and $L_2$ such that WFOMC of $(\sentence \land \gridaxiom(H,V), n^2)$ can be reduced to WFOMC of $(\sentence' \land \loaxiom(L_1) \land \loaxiom(L_2), n^2)$.
\end{lemma}

\begin{proof}[Proof of \Cref{lemma:grid-2lo}]
Our goal is to arrange the $n^2$ elements to a grid of $n$ rows and $n$ columns using $L_1$ and $L_2$ so that the horizontal and vertical successor relations $H,V$ satisfying \Cref{def:grid} can be defined.
To begin, we define the following useful relations by unary and binary relations using \Ctwo{} sentences for each linear order relation $L_i$ ($i \in \{1,2\}$):
\begin{itemize}
  \item The smallest element in the order, $First_i(x)$, by
  \begin{equation}\label{eq:2lo-first}
    \forall x \left( First_i(x) \leftrightarrow \forall y \ L_i(x,y) \right).
  \end{equation}
  \item The largest element in the order, $Last_i(x)$, by
  \begin{equation}\label{eq:2lo-last}
    \forall x \left( Last_i(x) \leftrightarrow \forall y \ L_i(y,x) \right).
  \end{equation}
  \item The successor relation of the order, $S_i(x,y)$, by
  \begin{equation}\label{eq:2lo-s}
    \begin{aligned}
      & \forall x (\lnot Last_i(x) \to \exists_{=1} y \ S_i(x,y)) \\
      \land & \forall x (\lnot First_i(x) \to \exists_{=1} y \ S_i(y,x)) \\
      \land & \forall x \forall y (S_i(x,y) \to L_i(x,y)),
    \end{aligned}
  \end{equation}
  which is a simplification of the trick in \cite[Lemma 2]{WFOMC-linearorder-axiom}.
\end{itemize}

\noindentparagraph{Building the Skeleton Using $L_1$}
We now build a basic skeleton of the grid by making constraints on $L_1$. We mark the elements at the left border and the right border by unary relations $Left$ and $Right$ with the following constraints:
\begin{itemize}
  \item Each border contains $n$ elements and the two borders are disjoint.
  \begin{equation}\label{eq:2lo-LeftRight1}
    \begin{aligned}
      & |Left| = |Right| = n \\
      \land & \forall x (\lnot Left(x) \lor \lnot Right(x)).
    \end{aligned}
  \end{equation}
  \item The first element of $L_1$ is at the left border and the last element is at the right border.
  \begin{equation}\label{eq:2lo-LeftRight2}
    \forall x ((First_1(x) \to Left(x)) \land (Last_1(x) \to Right(x))).
  \end{equation}
  \item $Right$ and $Left$ should be adjacent except for $First$ and $Last$.
  \begin{equation}\label{eq:2lo-LeftRight3}
    \forall x \forall y \left( S_1(x,y) \to (Right(x) \leftrightarrow Left(y)) \right).
  \end{equation}
\end{itemize}

Now we have exactly $n$ $Left$-$Right$ pairs in order. Each pair can be treated as an indicator of the start and the end of a row, therefore elements are arranged in the shape of \Cref{fig:L1-1}: There are $n$ rows starting with a $Left$ mark and ending with a $Right$ mark but each row can have an arbitrary length.

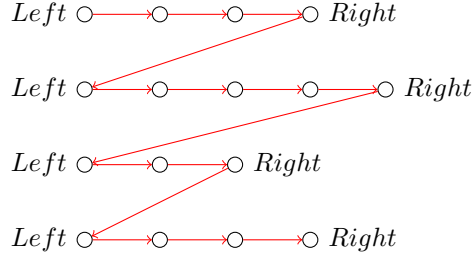
\begin{figure}
  \centering
      \begin{tikzpicture}
        \tikzstyle{roundnode}=[circle, draw, inner sep=0pt, minimum size=2mm]
        \tikzstyle{roundnodeleft}=[circle, draw, inner sep=0pt, minimum size=2mm, label=left:{\small $Left$}]
        \tikzstyle{roundnoderight}=[circle, draw, inner sep=0pt, minimum size=2mm, label=right:{\small $Right$}]
        \foreach \j in {0,3} {
            \node[roundnodeleft] (a0\j) at (0,\j) {};
            \foreach \i in {1,2} \node[roundnode] (a\i\j) at (\i,\j) {};
            \node[roundnoderight] (a3\j) at (3,\j) {};
            \foreach \i in {0,1,2}
                \pgfmathtruncatemacro{\k}{\i+1}
                \path[->,red] (a\i\j) edge (a\k\j);
        }
        \node[roundnodeleft] (a02) at (0,2) {};
        \node[roundnoderight] (a42) at (4,2) {};
        \foreach \i in {1,2,3} \node[roundnode] (a\i2) at (\i,2) {};
        \foreach \i in {0,1,2,3}
            \pgfmathtruncatemacro{\k}{\i+1}
            \path[->,red] (a\i2) edge (a\k2);
        \node[roundnodeleft] (a01) at (0,1) {};
        \node[roundnode] (a11) at (1,1) {};
        \node[roundnoderight] (a21) at (2,1) {};
        \path[->,red] (a01) edge (a11);
        \path[->,red] (a11) edge (a21);
        \path[->,red] (a33) edge (a02);
        \path[->,red] (a42) edge (a01);
        \path[->,red] (a21) edge (a00);
      \end{tikzpicture}
  \caption{A possible skeleton after making the constraints \Cref{eq:2lo-LeftRight1,eq:2lo-LeftRight2,eq:2lo-LeftRight3}, which consists of $n$ rows each starts with $Left$ and ends with $Right$. The red route represents $L_1$.}\label{fig:L1-1}
\end{figure}

We then mark nodes at the top and bottom by unary relations $Top$ and $Bottom$. If a node is marked with such a label, all nodes within the same row will have the same label. However, only the row containing $First_1$ can be $Top$ and only the row containing $Last_1$ can be $Bottom$. Therefore, only the first row is marked $Top$ and only the last row is marked $Bottom$. We further restrict the two rows to length $n$.
\begin{equation}\label{eq:TopBottom}
  \begin{aligned}
    & \forall x \forall y \Bigl((S_1(x,y) \land \lnot Right(x)) \\
    & \:\:\:\:\:\:\:\:\:\:\:\:\: \to \bigl((Top(x) \leftrightarrow Top(y)) \land (Bottom(x) \leftrightarrow Bottom(y))\bigl)\Bigl) \\
    \land & \forall x \bigl(First_1(x) \leftrightarrow (Left(x) \land Top(x))\bigl) \\
    \land & \forall x \bigl(Last_1(x) \leftrightarrow (Right(x) \land Bottom(x))\bigl) \\
    \land & |Top| = |Bottom| = n.
  \end{aligned}
\end{equation}

\noindentparagraph{Restricting the Shape Using $L_2$}
Next, we force the skeleton to be a grid by making constraints on $L_2$.

\begin{itemize}
  \item Both linear orders share the same starting and ending positions.
  \begin{equation}\label{eq:L2-1}
    \forall x ((First_1(x) \leftrightarrow First_2(x)) \land (Last_1(x) \leftrightarrow Last_2(x))).
  \end{equation}
  \item If $y$ is the $L_1$-successor of $x$ and $x$ is not at the right border, $x$ should be $L_2$-smaller than $y$.
  \begin{equation}\label{eq:L2-2}
    \forall x \forall y \bigl((S_1(x,y) \land \lnot Right(x)) \to L_2(x,y)\bigl).
  \end{equation}
  \item If $y$ is the $L_2$-successor of $x$ and $x$ is not at the bottom border, $x$ should be $L_1$-smaller than $y$.
  \begin{equation}\label{eq:L2-3}
    \forall x \forall y \bigl((S_2(x,y) \land \lnot Bottom(x)) \to L_1(x,y)\bigl).
  \end{equation}
  \item The $L_1$-successor and the $L_2$-successor of element $x$ are different.
  \begin{equation}\label{eq:L2-4}
    \forall x \forall y (\lnot S_1(x,y) \lor \lnot S_2(x,y)).
  \end{equation}
\end{itemize}

The following three observations imply the route of $L_2$ and the arrangement of elements.
\begin{enumerate}
  \item For any element $x$, the elements in the same row with $x$ but on its left should be $L_2$-smaller than $x$ due to \Cref{eq:L2-2}. In other words, each row will be visited by $L_2$ from left to right.
  \item For each $x$ not at the bottom, let $y$ be its successor in $L_2$. $y$ cannot lie in the rows above $x$ due to \Cref{eq:L2-3}, or in the same row with $x$ because otherwise by Observation (1) $y$ can only be the element next to $x$ at its right but that is against \Cref{eq:L2-4}.
  \item Based on the two observations above, the route of $L_2$ can be partitioned to several paths, each of which only descends until reaching the bottom. The length of each path is at most $n$ as there are only $n$ rows, and there can be at most $n$ paths since the bottom side has a length of $n$. Given that the size of the domain is $n^2$, there must be $n$ paths each of size $n$. Therefore, $L_2$ follows the blue route in \Cref{fig:L1L2} and the arrangement of elements forms a square grid.
\end{enumerate}

\begin{figure}
  \centering
      \begin{tikzpicture}
        \tikzstyle{roundnode}=[circle, draw, inner sep=0pt, minimum size=2mm]
        \foreach \i in {0,1,2,3}
            \foreach \j in {0,1,2,3}
                \node[roundnode] (a\i\j) at (\i,\j) {};
        \foreach \i in {0,1,2} {
            \foreach \j in {0,1,2,3}
                \pgfmathtruncatemacro{\k}{\i+1}
                \path[->,red] (a\i\j) edge (a\k\j);
            \pgfmathtruncatemacro{\k}{\i+1}
            \path[->,red] (a3\k) edge (a0\i);
        }
        \foreach \j in {0,1,2} {
            \foreach \i in {0,1,2,3}
                \pgfmathtruncatemacro{\k}{\j+1}
                \path[->,blue] (a\i\k) edge (a\i\j);
            \pgfmathtruncatemacro{\k}{\j+1}
            \path[->,blue] (a\j0) edge (a\k3);
        }
        \node[] (Top) at (1.5,3.5) {\small $Top$} ;
        \node[] (Bottom) at (1.5,-0.5) {\small $Bottom$} ;
        \node[] (Left) at (-0.5,1.5) {\small $Left$};
        \node[] (Right) at (3.6,1.5) {\small $Right$};
      \end{tikzpicture}
  \caption{The elements form a grid using linear orders $L_1$ (the red route) and $L_2$ (the blue route).}\label{fig:L1L2}
\end{figure}
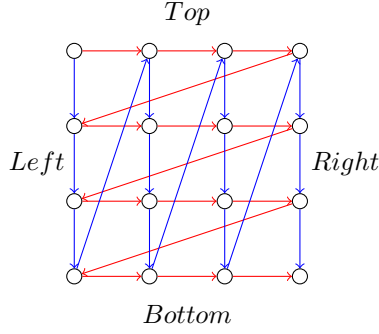

With the grid structure of elements, the horizontal and vertical successor relations can be defined as follows.
\begin{equation}\label{eq:get-hv}
  \begin{aligned}
    & \forall x \forall y \bigl(H(x,y) \leftrightarrow (S_1(x,y) \land \lnot Right(x))\bigl) \\
    \land & \forall x \forall y \bigl(V(x,y) \leftrightarrow (S_2(x,y) \land \lnot Bottom(x))\bigl).
  \end{aligned}
\end{equation}

It can be verified that \Cref{def:grid} is satisfied by assigning coordinates to elements based on the grid in Observation (3).

\noindentparagraph{Obtaining $\wfomc(\sentence, n^2)$}
Let $\sentence_l$ be the conjunction of $\sentence$ and \Cref{eq:2lo-first,eq:2lo-last,eq:2lo-s,eq:2lo-LeftRight1,eq:2lo-LeftRight2,eq:2lo-LeftRight3,eq:TopBottom,eq:L2-1,eq:L2-2,eq:L2-3,eq:L2-4,eq:get-hv}, which is a \Ctwo{} sentence with cardinality constraints. By the analysis above we obtain a unique grid from each model of $\sentence_l \land \loaxiom(L_1) \land \loaxiom(L_2)$. In addition, each assignment of the elements into a grid uniquely corresponds to a model of such a sentence. Therefore,
\begin{equation*}
  \begin{aligned}
    & \wfomc(\sentence \land \gridaxiom(H,V), n^2, \weight, \negweight)
    = \wfomc(\sentence_l \land \loaxiom(L_1) \land \loaxiom(L_2), n^2, \weight, \negweight)
  \end{aligned}
\end{equation*}
holds for any weighting functions $\weight, \negweight$.

By \Cref{lemma:c2+cc}, there is a sentence $\sentence'$ in \FOtwo{} without cardinality constraints such that WFOMC of $(\sentence_l \land \loaxiom(L_1) \land \loaxiom(L_2), n^2)$, as well as WFOMC of $(\sentence \land \gridaxiom(H,V), n^2)$, can be reduced to WFOMC of $(\sentence' \land \loaxiom(L_1) \land \loaxiom(L_2), n^2)$.
\end{proof}

We are now ready to state the main hardness result of this work.

\begin{theorem}\label{thm:2lo}
There is an \FOtwo{} sentence with two linear order relations whose WFOMC is \class{\#P_1}-hard to compute.
\end{theorem}

\begin{proof}
It follows naturally from \Cref{lemma:grid-2lo,lemma:grid-tm}.
\end{proof}

As pointed out in \cite{WFOMC-polys}, a linear order relation $L$ can be encoded in the following way: We first attach an acyclicity axiom to a binary relation $R$ and require $G(R)$ to be a tournament, and then let $L$ be the copy of $R$ but make it reflexive. Formally, a linear order relation $L$ can be encoded by the following sentence:
\begin{equation*}
  \begin{aligned}
    & \acyclicityaxiom(R) \land \forall x \forall y \bigl( (x \neq y) \to (R(x,y) \lor R(y,x)) \bigl) \\
    \land & \forall x \forall y \bigl( (x \neq y) \to (L(x,y) \leftrightarrow R(x,y)) \bigl) \\
    \land & \forall x \ L(x,x).
  \end{aligned}
\end{equation*}

The equality symbol can be eliminated by the trick in \cite[Lemma 3.5]{WFOMC-FO3} and thus we obtain the hardness for \FOtwo{} with two acyclic relations.

\begin{corollary}
There is an \FOtwo{} sentence with two acyclic relations whose WFOMC is \class{\#P_1}-hard to compute.
\end{corollary} 

\section{Domain-Liftability of A Linear Order Axiom and A Successor Axiom}

In addition to the hardness of WFOMC of \FOtwo{} with two linear order relations, we show that weakening the power of one of the two linear order axioms by restricting access to only its successor relation results in domain-liftability. We give an algorithm which computes WFOMC of \FOtwo{} and \Ctwo{} sentences with a linear order relation (possibly with its successor relation) and another successor relation in time polynomial in the domain size.

\subsection{Warmup: Framework of Computing WFOMC for \FOtwo{}}\label{sec:framework}

We briefly introduce the existing framework of computing WFOMC for \FOtwo{} sentences and important concepts proposed in \cite{WFOMC-FO2,WFOMC-linearorder-axiom} which are used in our algorithm.

A set of literals is \emph{maximally consistent} if it does not contain both a literal and its negation, and cannot be extended to a larger consistent set.

\begin{definition}[1-type]
A \emph{1-type} of a first-order sentence $\sentence$ is a maximally consistent set of literals formed from relations in $\sentence$ where each literal uses only a single variable $x$.
\end{definition}

\begin{definition}[2-table]
A \emph{2-table} of a first-order sentence $\sentence$ is a maximally consistent set of literals formed from relations in $\sentence$ where each literal uses two variables $x, y$.
\end{definition}


For example, $\sentence = \forall x \forall y \left( F(x,y) \to G(y) \right)$ has four 1-types: $F(x,x) \land G(x)$, $F(x,x) \land \lnot G(x)$, $\lnot F(x,x) \land G(x)$ and $\lnot F(x,x) \land \lnot G(x)$, and four 2-tables: $F(x,y) \land F(y,x)$, $F(x,y) \land \lnot F(y,x)$, $\lnot F(x,y) \land F(y,x)$ and $\lnot F(x,y) \land \lnot F(y,x)$. Intuitively, a 1-type interprets unary and reflexive binary relations for a single domain element, and a 2-table interprets binary relations for a pair of distinct domain elements.

Given a WFOMC problem $\wfomc(\sentence, n, \weight, \negweight)$ for an \FOtwo{} sentence $\sentence$, the Skolemization trick in \cite{WFOMC-FO2} produces a universally quantified \FOtwo{} sentence $\sentence'$ of size $O(|\sentence|)$ and a pair of weighting functions $\weight', \negweight'$ such that $\wfomc(\sentence, n, \weight, \negweight) = \wfomc(\sentence', n, \weight', \negweight')$. Specifically, $\sentence'$ is in the form $\forall x \forall y \ \psi(x,y)$ where $\psi(x,y)$ is a quantifier-free \FOtwo{} sentence.
Therefore, $\sentence'$ can be expanded as the conjunction of ground formulae over the domain $[n]$:
\begin{equation*}
  \begin{aligned}
    \sentence' = \left( \bigwedge_{i=1}^n \psi(i,i) \right) \land \left( \bigwedge_{1 \le i < j \le n} \psi(i,j) \land \psi(j,i) \right).
  \end{aligned}
\end{equation*}

Let $C = \{C_1, C_2, \cdots, C_u\}$ be the set of possible 1-types of $\sentence'$ and $D$ be the set of possible 2-tables. Suppose that the 1-type of element $i$ is determined as $C_{\tau_i}$ ($\tau_i \in [u]$). After substituting unary and reflexive binary literals in $\psi(i,j) \land \psi(j,i)$ with true or false according to $C_{\tau_i}$ and $C_{\tau_j}$, the formula for a pair of elements $(i,j)$ does not have common ground literals with any other pair, hence the 2-tables between each pair of elements can be selected independently. Moreover, the formula only depends on $C_{\tau_i}$ and $C_{\tau_j}$. Therefore, the WFOMC can be computed as follows:
\begin{equation}\label{eq:basicwfomc}
  \begin{aligned}
    & \wfomc(\sentence', n, \weight', \negweight')
    = \sum_{\tau_1, \cdots, \tau_n \in [u]} \ \prod_{i=1}^n W(C_{\tau_i}) \prod_{1 \le i < j \le n} r_{\tau_i,\tau_j},
  \end{aligned}
\end{equation}
where
\begin{equation*}
  r_{s,t} = \sum_{\substack{\pi \in D, \\ C_{s}(a) \land C_{t}(b) \land \pi(a,b) \models \psi(a,b) \land \psi(b,a)}} W(\pi).
\end{equation*}

We follow the idea from \cite{WFOMC-linearorder-axiom} that computes \Cref{eq:basicwfomc} recursively. Define
\begin{equation}\label{eq:f-def}
  f(m, \tau_1, \cdots, \tau_m) = \prod_{i=1}^m W(C_{\tau_i}) \prod_{1 \le i < j \le m} r_{\tau_i,\tau_j}
\end{equation}
for each $m \in [n]$ and $\tau_1, \cdots, \tau_m \in [u]$. Then for $m \ge 2$, it holds that
\begin{equation}\label{eq:f-recursion}
  \begin{aligned}
  & f(m, \tau_1, \cdots, \tau_m)
  = f(m-1, \tau_1, \cdots, \tau_{m-1}) \cdot W(C_{\tau_m}) \cdot \prod_{i=1}^{m-1} r_{\tau_i,\tau_{m}},
  \end{aligned}
\end{equation}
and the WFOMC can be obtained by
\begin{equation}\label{eq:f-answer}
  \wfomc(\sentence', n, \weight', \negweight') = \sum_{\tau_1, \cdots, \tau_n \in [u]} f(n,\tau_1, \cdots, \tau_n).
\end{equation}

The computation can be further adapted to an algorithm that runs in time polynomial in $n$ but we stop here as it suffices to point out the critical idea of recursion.

\subsection{The Algorithm}

Now we describe a new algorithm for WFOMC for \FOtwo{} sentences with a linear order relation and a successor relation (of another linear order).

\begin{theorem}\label{thm:lo+succ}
An \FOtwo{} sentence $\sentence$ with a linear order relation and another successor relation is domain-liftable.
\end{theorem}

\begin{proof}
We obtain $\sentence', \weight', \negweight'$ by the transformation mentioned in the previous subsection and then our task is to compute $\wfomc(\sentence' \land \loaxiom(L) \land \succaxiom(S), n, \weight', \negweight')$ where $L, S$ are distinguished binary relations. Again, let $C = \{C_1, C_2, \cdots, C_u\}$ be the set of possible 1-types of $\sentence'$, $D$ be the set of possible 2-tables and $C_{\tau_i}$ ($\tau_i \in [u]$) be the 1-type of element $i$.

Each of the two axioms implies an order of elements. We can, without loss of generality, fix the order of $L$ to be $1 \to 2 \to 3 \to \cdots \to n$. For any other possible order of $L$, applying a permutation on the indices of elements maps the models to those for the fixed $L$ order while preserving the weight. This fact is formally stated by \cite[Theorem 1]{WFOMC-linearorder-axiom}.

\noindentparagraph{Computing WFOMC for a Fixed $S$-Order}
As a warm-up, we also fix the order of $S$ as a permutation of $[n]$. We compute the WFOMC recursively following the idea of \Cref{eq:f-def,eq:f-recursion,eq:f-answer}. For each $s,t \in [u]$ and $k \in \{1,2,3\}$, define
\begin{equation*}
  r_{s,t,k} = \sum_{\substack{\pi \in D, \\ C_{s}(a) \land C_{t}(b) \land \pi(a,b) \models \psi(a,b) \land \psi(b,a) \land \phi_k(a,b)}} W(\pi),
\end{equation*}
where
\begin{equation*}
  \begin{aligned}
    \phi_1(a,b) &= L(a,b) \land \lnot L(b,a) \land \lnot S(a,b) \land \lnot S(b,a), \\
    \phi_2(a,b) &= L(a,b) \land \lnot L(b,a) \land S(a,b) \land \lnot S(b,a), \\
    \phi_3(a,b) &= L(a,b) \land \lnot L(b,a) \land \lnot S(a,b) \land S(b,a).
  \end{aligned}
\end{equation*}

The term $r_{s,t,k}$ is a refinement of $r_{\tau_i,\tau_j}$ in \Cref{eq:f-def} since $L$ and $S$ are totally interpreted and therefore $r_{\tau_i,\tau_j}$ can only take one of the three values from $\{r_{\tau_i,\tau_j,1}, r_{\tau_i,\tau_j,2}, r_{\tau_i,\tau_j,3}\}$. Note that we always compute $r_{\tau_i, \tau_j}$ for those $i<j$ and we fix the order of $L$ as the natural order, hence $\phi_k(a,b)$ does not involve the cases in which $\lnot L(a,b) \land L(b,a)$ holds.

The function $f$ in \Cref{eq:f-def} can be adapted as follows:
\begin{equation}\label{eq:g-def}
  g(m, \tau_1, \cdots, \tau_m) = \prod_{i=1}^m W(C_{\tau_i}) \prod_{1 \le i < j \le m} r_{\tau_i,\tau_j,\kappa_{i,j}},
\end{equation}
where each $\kappa_{i,j}$ takes the only value from $\{1,2,3\}$ such that $\phi_{\kappa_{i,j}}(i,j)$ is satisfied. We then obtain the recursive computation for $g$:
\begin{equation}\label{eq:g-recursion}
  \begin{aligned}
  & g(m, \tau_1, \cdots, \tau_m)
  = g(m-1, \tau_1, \cdots, \tau_{m-1}) \cdot W(C_{\tau_m}) \cdot \prod_{i=1}^{m-1} r_{\tau_i,\tau_{m},\kappa_{i,m}}.
  \end{aligned}
\end{equation}

\noindentparagraph{Computing WFOMC for All $S$-Orders}
The computation above has two main drawbacks: it computes WFOMC for a fixed order of $S$, and the concrete 1-type sequence $\tau_1, \cdots, \tau_m$ is used to determine the weights. These two issues might lead to $n! \cdot u^n$ possibilities of $g$-values to be computed. The key to addressing these issues is to figure out how the terms $r_{\tau_i,\tau_{m},\kappa_{i,m}}$ can appear in \Cref{eq:g-recursion}.
This relies on the following observation.

Imagine that we place the elements $1, 2, \cdots, n$ in a line in order from left to right. The order implied by $S$ is like a string of the $n$ elements. If we focus on the first $m$ elements, ignoring other elements, the string breaks to several \emph{segments}. Formally, a segment is a maximal sequence of elements $x_1, x_2, \cdots, x_k$ such that every element in the sequence is within the first $m$ elements of the order of $L$, and $S(x_1, x_2) \land \cdots \land S(x_{k-1}, x_k)$ holds. The element $x_1$ is called the head of the segment, and $x_k$ is called the tail of the segment. An example is shown in \Cref{fig:segments} where we consider the order of $S$ as $5 \to 2 \to 8 \to 1 \to 6 \to 4 \to 7 \to 10 \to 3 \to 9$. Looking at only the prefix $\{1,2,\cdots,6\}$, there are three segments $5 \to 2$, $1 \to 6 \to 4$ and $3$.

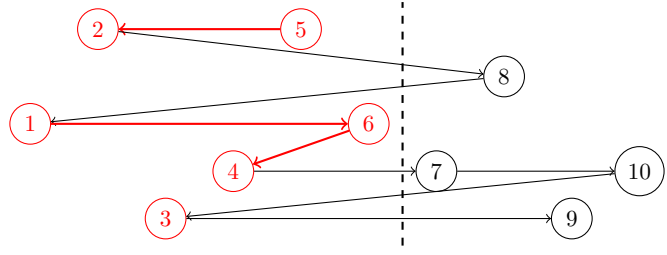
\begin{figure}
  \centering
      \begin{tikzpicture}[scale=0.9]
        \tikzstyle{vertex}=[circle,draw, scale=0.8, minimum size=3mm]
        \tikzstyle{redvertex}=[circle,draw,red, scale=0.8, minimum size=3mm]
        \node[redvertex] (5) at (5,2.8) {$5$};
        \node[redvertex] (2) at (2,2.8) {$2$};
        \node[vertex] (8) at (8,2.1) {$8$};
        \node[redvertex] (1) at (1,1.4) {$1$};
        \node[redvertex] (6) at (6,1.4) {$6$};
        \node[redvertex] (4) at (4,0.7) {$4$};
        \node[vertex] (7) at (7,0.7) {$7$};
        \node[vertex] (10) at (10,0.7) {$10$};
        \node[redvertex] (3) at (3,0) {$3$};
        \node[vertex] (9) at (9,0) {$9$};
        \path[->,red,thick] (5) edge (2);
        \path[->] (2) edge (8);
        \path[->] (8) edge (1);
        \path[->,red,thick] (1) edge (6);
        \path[->,red,thick] (6) edge (4);
        \path[->] (4) edge (7);
        \path[->] (7) edge (10);
        \path[->] (10) edge (3);
        \path[->] (3) edge (9);
        \draw[dashed,thick] (6.5,-0.4)--(6.5,3.2);
      \end{tikzpicture}
  \caption{An illustration of the segments (red nodes and red arrows) of the order of $S$ ($5 \to 2 \to 8 \to 1 \to 6 \to 4 \to 7 \to 10 \to 3 \to 9$) on a prefix of elements $\{1, 2, \cdots, 6\}$.}\label{fig:segments}
\end{figure}


Denote a segment with the head having 1-type $C_s$ and the tail having 1-type $C_t$ by $\segment{s}{t}$. Let $\bar \otc = (\bar \otce_1, \cdots, \bar \otce_u)$ be the vector of length $u$ where $\bar \otce_i$ is the number of elements $k \in [1,m-1]$ such that $\tau_k = i$. Observe that the value of $\lambda = \prod_{i=1}^{m-1} r_{\tau_i,\tau_{m},\kappa_{i,m}}$ can be classified to the following four types depending on the connection between $m$ and the former $m-1$ elements with respect to $S$:

\begin{itemize}
  \item $m$ merges two segments $\segment{a}{b}$ and $\segment{c}{d}$ by linking the tail of $\segment{a}{b}$ to $m$ and linking $m$ to the head of $\segment{c}{d}$. If $b \neq c$, then $\lambda$ takes the value
  \begin{equation*}
    \begin{aligned}
    \lambda_{\text{merge1}}(b,c,\tau_m,\bar \otc) =\ & r_{b,\tau_m,2} \cdot r_{c,\tau_m,3} \cdot \left(r_{b,\tau_m,1}\right)^{\bar \otce_b-1}
     \cdot \left(r_{c,\tau_m,1}\right)^{\bar \otce_c-1} \cdot \prod_{\substack{s \in [u],\\ s \neq b, \\ s \neq c}} \left(r_{s,\tau_m,1}\right)^{\bar \otce_s}.
    \end{aligned}
  \end{equation*}
  If $b=c$ (i.e., if the 1-types of the tail of the first segment and the head of the second segment are the same), then $\lambda$ takes the value
  \begin{equation*}
    \begin{aligned}
    \lambda_{\text{merge2}}(b,\tau_m,\bar \otc) =\ & r_{b,\tau_m,2} \cdot r_{b,\tau_m,3}
     \cdot \left(r_{b,\tau_m,1}\right)^{\bar \otce_b-2}
    \cdot \prod_{\substack{s \in [u],\\ s \neq b}} \left(r_{s,\tau_m,1}\right)^{\bar \otce_s}.
    \end{aligned}
  \end{equation*}
  \item $m$ extends a segment $\segment{a}{b}$ serving as a new head. In this case, $\lambda$ takes the value
  \begin{equation*}
    \lambda_{\text{head}}(a,\tau_m,\bar \otc) =\ r_{a,\tau_m,3} \cdot \left(r_{a,\tau_m,1}\right)^{\bar \otce_a-1} \cdot \prod_{\substack{s \in [u],\\ s \neq a}} \left(r_{s,\tau_m,1}\right)^{\bar \otce_s}.
  \end{equation*}
  \item $m$ extends a segment $\segment{a}{b}$ serving as a new tail. In this case, $\lambda$ takes the value
  \begin{equation*}
    \lambda_{\text{tail}}(b,\tau_m,\bar \otc) =\ r_{b,\tau_m,2} \cdot \left(r_{b,\tau_m,1}\right)^{\bar \otce_b-1} \cdot \prod_{\substack{s \in [u],\\ s \neq b}} \left(r_{s,\tau_m,1}\right)^{\bar \otce_s}.
  \end{equation*}
  \item $m$ creates a new segment containing itself only. In this case, $\lambda$ takes the value
  \begin{equation*}
    \lambda_{\text{only}}(\tau_m,\bar \otc) = \prod_{s \in [u]} \left(r_{s,\tau_m,1}\right)^{\bar \otce_s}.
  \end{equation*}
\end{itemize}

The computation of $\lambda$ suggests that we no longer need to care about the concrete order of $S$ and the 1-type sequence $\tau_1, \cdots, \tau_m$, but rather $\bar \otc$, $\tau_m$, the behavior of $m$ with respect to $S$, and the number of heads and tails of segments assigned to each 1-type. Define $\otc$ similarly as $\bar \otc$ but within elements $1, \cdots, m$, and let $\segc = (\segce_{1,1}, \cdots, \segce_{u,u})$ be the vector of length $u^2$ where $\segce_{s,t}$ is the number of segments $\segment{s}{t}$ within elements $1, \cdots, m$. Now we can adapt \Cref{eq:g-def} by grouping orders of $S$ and the possible 1-type sequences $\tau_1, \cdots, \tau_m$ that have the same $\otc$ and $\segc$. Define $h(m,\otc, \segc)$ as
\begin{equation*}
  h(m, \otc, \segc) = \sum_{\substack{\tau_1, \cdots, \tau_m \text{ satisfying } \otc,\\ \text{segments of } 1, \cdots, m \\ \text{satisfying } \segc}} \  \prod_{i=1}^m W(C_{\tau_i}) \prod_{1 \le i < j \le m} r_{\tau_i,\tau_j,\kappa_{i,j}}.
\end{equation*}

Note that by the analysis above, the value of $\prod_{1 \le i < j \le m} r_{\tau_i,\tau_j,\kappa_{i,j}}$ is also unique given $\tau_1, \cdots, \tau_m$ and the segments of $1, \cdots, m$.

The value of $h$ for $m$ elements can be computed from $h$ for $m-1$ elements by enumerating the 1-type $C_{\tau}$ of $m$ and the behavior of $m$ with respect to $S$.
Denote by $\otc^{-\tau}$ the vector of size $u$ such that
\begin{equation*}
  \otce^{-\tau}_{s} = \begin{cases}
    \otce_{s}-1, & s=\tau, \\
    \otce_{s}, & \mbox{otherwise.}
  \end{cases}
\end{equation*}

Similarly, $\segc^{-(a,b)}$ denotes the vector of size $u^2$ such that
\begin{equation*}
  \segce^{-(a,b)}_{s,t} = \begin{cases}
    \segce_{s,t} - 1, & s=a,\ t=b, \\
    \segce_{s,t}, & \mbox{otherwise,}
  \end{cases}
\end{equation*}
and $\segc^{+(a,b)}$ is defined similarly. The operation can be nested, e.g., $\segc^{-(a,b)+(c,d)}$ refers to $\left(\segc^{-(a,b)}\right)^{+(c,d)}$.

We compute $h(m,\otc,\segc)$ given the 1-type $C_{\tau}$ of $m$ and the behavior $\beta\in \{\text{merge1},\text{merge2},\text{head},\text{tail},\text{only}\}$ of $m$, denoted by $h(m,\otc,\segc | \tau,\beta)$:
\begin{itemize}
  \item $\beta=\text{merge1}$: $m$ merges two segments $\segment{a}{b}$ and $\segment{c}{d}$ by linking the tail of $\segment{a}{b}$ to $m$ and linking $m$ to the head of $\segment{c}{d}$ for some $a,b,c,d \in [u]$ and $b \neq c$. Elements $1, \cdots, m-1$ have 1-types consistent with $\otc^{-\tau}$ and segments consistent with $\segc^{\text{merge1}} = \segc^{-(a,d)+(a,b)+(c,d)}$. Then,
      \begin{equation*}
        \begin{aligned}
          & h(m,\otc,\segc | \tau,\text{merge1})
           = \sum_{\substack{a,b,c,d \in [u],\ b \neq c, \\ \segc^{\text{merge1}} \ge 0}} \Bigl(h(m-1, \otc^{-\tau}, \segc^{\text{merge1}}) \cdot \eta \cdot W(C_{\tau}) \cdot \lambda_{\text{merge1}}(b,c,\tau,\otc^{-\tau})\Bigl),
        \end{aligned}
      \end{equation*}
      where
      \begin{equation*}
        \eta = \begin{cases}
          \segce^{\text{merge1}}_{a,b} \cdot \segce^{\text{merge1}}_{c,d}, & a \neq c \text{ or } b \neq d, \\
          \segce^{\text{merge1}}_{a,b}\left(\segce^{\text{merge1}}_{a,b}-1\right), & \mbox{otherwise.}
        \end{cases}
      \end{equation*}
  \item $\beta=\text{merge2}$: $m$ merges two segments $\segment{a}{b}$ and $\segment{b}{d}$ by linking the tail of $\segment{a}{b}$ to $m$ and linking $m$ to the head of $\segment{b}{d}$ for some $a,b,d \in [u]$. Elements $1, \cdots, m-1$ have 1-types consistent with $\otc^{-\tau}$ and segments consistent with $\segc^{\text{merge2}} = \segc^{-(a,d)+(a,b)+(b,d)}$. Then,
      \begin{equation*}
        \begin{aligned}
          & h(m,\otc,\segc | \tau,\text{merge2})
           = \sum_{a,b,d \in [u],\ \segc^{\text{merge2}} \ge 0} \Bigl(h(m-1, \otc^{-\tau}, \segc^{\text{merge2}}) \cdot \eta \cdot W(C_{\tau}) \cdot \lambda_{\text{merge2}}(b,\tau,\otc^{-\tau})\Bigl),
        \end{aligned}
      \end{equation*}
      where
      \begin{equation*}
        \eta = \begin{cases}
          \segce^{\text{merge2}}_{a,b} \cdot \segce^{\text{merge2}}_{b,d}, & a \neq b \text{ or } b \neq d, \\
          \segce^{\text{merge2}}_{a,b}\left(\segce^{\text{merge2}}_{a,b}-1\right), & \mbox{otherwise.}
        \end{cases}
      \end{equation*}
  \item $\beta=\text{head}$: $m$ extends a segment $\segment{a}{b}$ serving as a new head for some $a,b \in [u]$. Elements $1, \cdots, m-1$ have 1-types consistent with $\otc^{-\tau}$ and segments consistent with $\segc^{\text{head}} = \segc^{-(\tau,b)+(a,b)}$. Then,
      \begin{equation*}
        \begin{aligned}
          h(m,\otc,\segc | \tau,\text{head})
          =&  \sum_{a,b \in [u],\ \segc^{\text{head}} \ge 0} \Bigl(h(m-1, \otc^{-\tau}, \segc^{\text{head}}) \cdot \segce^{\text{head}}_{a,b} \cdot W(C_{\tau}) \cdot \lambda_{\text{head}}(a,\tau,\otc^{-\tau})\Bigl).
        \end{aligned}
      \end{equation*}
  \item $\beta=\text{tail}$: $m$ extends a segment $\segment{a}{b}$ serving as a new tail for some $a,b \in [u]$. Elements $1, \cdots, m-1$ have 1-types consistent with $\otc^{-\tau}$ and segments consistent with $\segc^{\text{tail}} = \segc^{-(a,\tau)+(a,b)}$. Then,
      \begin{equation*}
        \begin{aligned}
          h(m,\otc,\segc | \tau,\text{tail})
          =&  \sum_{a,b \in [u],\ \segc^{\text{tail}} \ge 0} \Bigl(h(m-1, \otc^{-\tau}, \segc^{\text{tail}}) \cdot \segce^{\text{tail}}_{a,b} \cdot W(C_{\tau}) \cdot \lambda_{\text{tail}}(b,\tau,\otc^{-\tau})\Bigl).
        \end{aligned}
      \end{equation*}
  \item $\beta=\text{only}$: $m$ creates a new segment containing itself only. Elements $1, \cdots, m-1$ have 1-types consistent with $\otc^{-\tau}$ and segments consistent with $\segc^{\text{only}} = \segc^{-(\tau,\tau)}$. Then,
      \begin{equation*}
        \begin{aligned}
          &h(m,\otc,\segc | \tau,\text{only})
          =  h(m-1, \otc^{-\tau}, \segc^{\text{only}}) \cdot W(C_{\tau}) \cdot \lambda_{\text{only}}(\tau,\otc^{-\tau}).
        \end{aligned}
      \end{equation*}
\end{itemize}

Summing the above terms up, we have
\begin{equation*}
  \begin{aligned}
    h(m, \otc, \segc)
    = \sum_{\substack{\tau \in [u],\\ \otce_{\tau}>0}} \ \sum_{\substack{\beta\in \{\text{merge1},\text{merge2},\\ \text{head},\text{tail},\text{only}\}}} h(m,\otc,\segc | \tau,\beta).
  \end{aligned}
\end{equation*}

Finally, the WFOMC for the fixed order of $L$ can be obtained when each of the $n$ elements is assigned some 1-type and there are only $1$ segment:
\begin{equation*}
  \gamma = \sum_{|\otc| = n} \sum_{|\segc| = 1} h(n, \otc, \segc).
\end{equation*}

As mentioned at the beginning of the proof, every other order of $L$ has the same weight as the one with the natural order of numbers. Therefore,
\begin{equation*}
  \wfomc(\sentence' \land \loaxiom(L) \land \succaxiom(S), n, \weight', \negweight') = n! \cdot \gamma.
\end{equation*}

The algorithm is summarized in \Cref{algo:2succ}.

\begin{algorithm}
\caption{WFOMC for \FOtwo{}+$\loaxiom$+$\succaxiom$} \label{algo:2succ}
\KwIn{An \FOtwo{} sentence $\sentence$ with distinguished binary relations $L, S$, weighting functions $\weight, \negweight$, and an integer $n$}
\KwOut{$\wfomc(\sentence \land \loaxiom(L) \land \succaxiom(S), n, \weight, \negweight)$}

$h_{1, \otc,\segc} \gets W(\tau) \textbf{ for each } \tau \in [u], \ \otce_{\tau} = \segce_{\tau,\tau} = |\otc| = |\segc| = 1$ \\
\For{$m \gets 2$ \KwTo $n$} {
    \ForEach{\upshape $\otc, \segc$ \textbf{ such that } $|\otc|=m$} {
        $h(m, \otc, \segc) \gets 0$ \\
        \ForEach{\upshape $\tau \in [u]$ \textbf{ such that } $\otce_{\tau>0}$} {
            \ForEach{\upshape $\beta\in \{\text{merge1},\text{merge2},\text{head},\text{tail},\text{only}\}$} {
                compute $h(m,\otc,\segc | \tau,\beta)$ \\
                $h(m, \otc, \segc) \gets h(m, \otc, \segc) + h(m,\otc,\segc | \tau,\beta)$
            }
        }
    }
}
$\gamma = \sum_{|\otc| = n} \sum_{|\segc| = 1} h(n,\otc,\segc)$ \\
\Return $n! \cdot \gamma$
\end{algorithm}

As the number of $h$-values we need to compute is $n \cdot O(n^u) \cdot O(n^{u^2}) = O(n^{O(u^2)})$ and each computation takes time polynomial in $u$, the algorithm runs in time $O(n^{O(u^2)})$, thus the fragment is domain-liftable.
\end{proof}

The result can be extended to the fragment of \Ctwo{} with cardinality constraints by \Cref{lemma:c2+cc}.

\begin{corollary}
\Ctwo{} with a linear order relation and another successor relation (possibly along with cardinality constraints) is domain-liftable for WFOMC.
\end{corollary}

Using \Cref{eq:2lo-s}, one can derive the successor relation of a linear order relation by a simple \Ctwo{} sentence, therefore the linear order relation in \Cref{thm:2succ} can be along with its successor relation.

\begin{corollary}\label{thm:2succ}
\Ctwo{} with a linear order relation, its successor relation and another successor relation (possibly along with cardinality constraints) is domain-liftable for WFOMC.
\end{corollary}




\section{Performance Measurements}

We implemented the Algorithm \ref{algo:2succ} in Python\footnote{The implementation can be found in \url{https://github.com/kulavacl/WFOMC}}.
Although this implementation is not optimized for performance, it provides valuable insight into how the algorithm scales with the size of the problem.

\subsection{Selected Sentences}

We first selected four sentences with 1, 2, 4, and 5 valid 1-types\footnote{A 1-type is \emph{valid} if it is satisfiable in some model of the sentence. It is a natural idea to remove unsatisfiable 1-types in the implementation.}, respectively, to evaluate the performance of Algorithm \ref{algo:2succ}. Since the complexity of our algorithm is exponential in the number of valid 1-types of the sentence and polynomial in the size of the domain, benchmarking the algorithm on formulas with different number of valid 1-types helps us better understand its scalability.
All the four sentences contain distinguished binary relations $L$ and $S$. The linear order axiom will be applied to $L$ and the successor axiom will be applied to $S$, that is, $\loaxiom(L) \land \succaxiom(S)$ will be added to the sentences implicitly.

\begin{itemize}
  \item $\Phi_1$ is a sentence that has a single valid 1-type:
    \begin{align*}
        \Phi_1 = \forall x \forall y \Bigl( & (B(x, y) \to S(x, y)) \\
        & \wedge \bigl((S(x, y) \wedge L(x, y)) \to B(x, y)\bigl)\Bigl).
    \end{align*}
  \item $\Phi_2$ is modification of $\Phi_1$ with 2 valid 1-types:
    \begin{align*}
        \Phi_2 = \forall x \forall y \bigl((S(x, y) \wedge L(x, y)) \to (U(x) \leftrightarrow \lnot U(y))\bigl).
    \end{align*}
  \item $\Phi_4$ is a sentence with 4 valid 1-types:
    \begin{align*}
        \Phi_4 = \forall x \forall y \Bigl( & \bigl((U_1(x) \wedge L(x, y)) \to U_1(y)\bigl) \\
        & \land \bigl((U_1(x) \wedge S(x, y)) \to U_2(y)\bigl) \Bigl).
    \end{align*}
  \item $\Phi_5$ is an extension of $\Phi_4$ which has 5 valid 1-types:
    \begin{align*}
        \Phi_5 =  \Phi_4 \land \forall x \forall y \bigl(B(x, y) \to (U_1(x) \wedge U_2(y))\bigl).
    \end{align*}
\end{itemize}

We modified the algorithm to only consider one fixed linear order of $L$ rather than computing the result over all possible linear orders. This allows us to match the sequences generated by WFOMC of $\Phi_1$ to the existing sequence A000670 in OEIS\footnote{http://oeis.org/}, and $\Phi_2$ to A000629, which will be proved later. There is no matching from $\Phi_4$ or $\Phi_5$ to any sequence in OEIS but we still use them since the number of valid 1-types is the deciding factor in the scalability of our algorithm.


Another sentence on which we evaluate the performance is the formalization of a real-world scenario. Suppose there are $n$ train stations in a single line, where $L(x, y)$ signifies that station $x$ is to the west of station $y$. A train is to stop at each station exactly once, in an order specified by a successor relation $S$ of all stations. We also require that the train starts in the west-most station and ends its journey in the east-most station. Reverting the direction of the train is costly, so we limit the number of times the train has to turn around to at most $m$.

The scenario can be formalized by the following sentence $\Phi_{train}$. Each element in the domain represents a station. The westmost station and the eastmost station are labelled by the unary relation $First$ and $Last$, respectively. The unary relation $W2E(x)$ being true indicates that the train goes to $x$ from west to east, and being false indicates the reversed direction. The unary relation $RevertAt(x)$ indicates whether the train reverts at the station $x$. Specifically, we let $W2E(x)$ be true if $x$ is the westmost station, and let $RevertAt(x)$ be false if $x$ is the eastmost station.

\begin{align*}
    \Phi_{train} = & \ \ \ \ \forall x (First(x) \leftrightarrow \lnot \exists y \ S(y, x)) \\
    & \land \forall x (Last(x) \leftrightarrow \lnot \exists y \ S(x, y)) \\
    & \land \forall x (First(x) \to W2E(x)) \\
    & \land \forall x (Last(x) \to \lnot RevertAt(x)) \\
    & \wedge \forall x \forall y \bigl(S(x,y) \to (W2E(y) \leftrightarrow L(x,y))\bigl) \\
    & \wedge \forall x \forall y \Bigl(S(x,y) \to \bigl(RevertAt(x) \leftrightarrow (W2E(x) \leftrightarrow \lnot W2E(y))\bigl)\Bigl) \\
    & \land |RevertAt| \leq m.
\end{align*}

In our experiment, we chose $m = 2$, for which, if the linear order is fixed, the sequence of counting the number of models corresponds to sequence A111277 in OEIS.

We remark that the nature of our algorithm allows us to obtain the first and last element in the linear order efficiently.
This means that while $\Phi_{train}$ has 6 valid 1-types, the main computation only needs to iterate over 4 of them.

\subsection{Performance}

We compared the runtimes of Algorithm \ref{algo:2succ} to runtimes of weighted model counters d4\footnote{https://github.com/crillab/d4} and GANAK \cite{ganak}.
For formulas $\Phi_2$, $\Phi_4$, and $\Phi_5$, we stopped the execution once a single instance exceeded one hour of computation for both our algorithm and the model counters.
For formula $\Phi_1$, the model counters were stopped once a single instance of the problem took 600 seconds to solve, while our algorithm was limited to domain size 2000, which is enough to showcase its much better performance for this formula.

All measurements were taken on a laptop with an Intel Core i7-10510U CPU with 8 cores, each running at 1.8 GHz and 32GB of RAM.

The resulting runtimes can be found in Figure \ref{fig:execution}. In all cases, the logarithm of the running time exhibits a bend as the domain size $n$ grows. This behavior provides evidence that the runtime of our algorithm is indeed in polynomial of $n$. Moreover, our algorithm scales better than weighted model counters when the domain size is large.

\begin{figure}
  \centering
  \begin{subfigure}[t]{.32\textwidth}
    \centering
    \includegraphics[width=\linewidth]{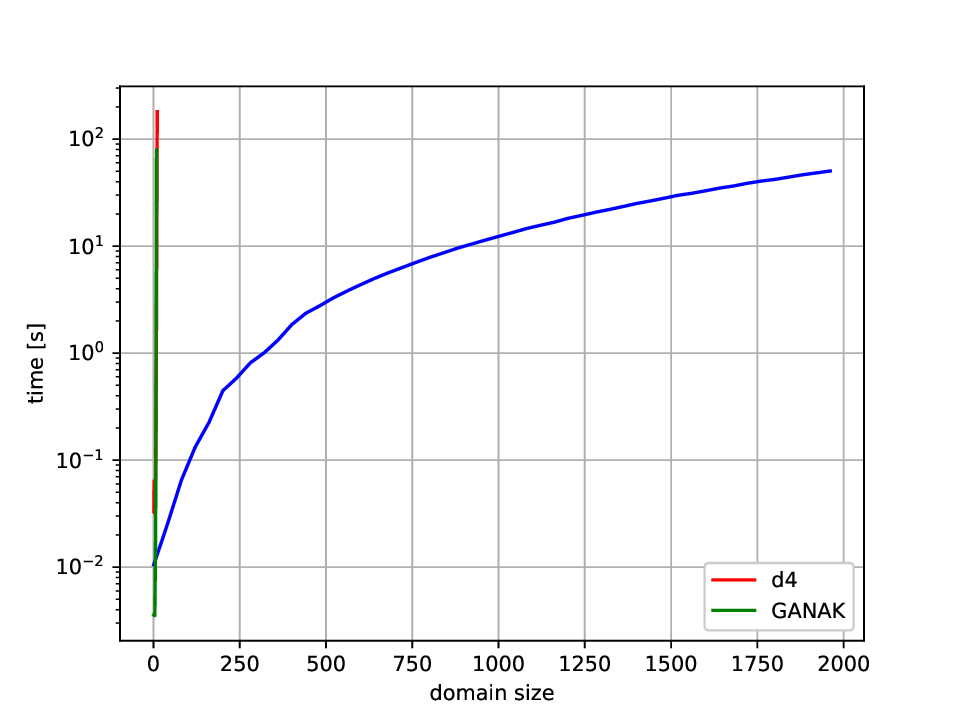}
    \caption{Runtime for $\Phi_1$}
  \end{subfigure}
  \begin{subfigure}[t]{.32\textwidth}
    \centering
    \includegraphics[width=\linewidth]{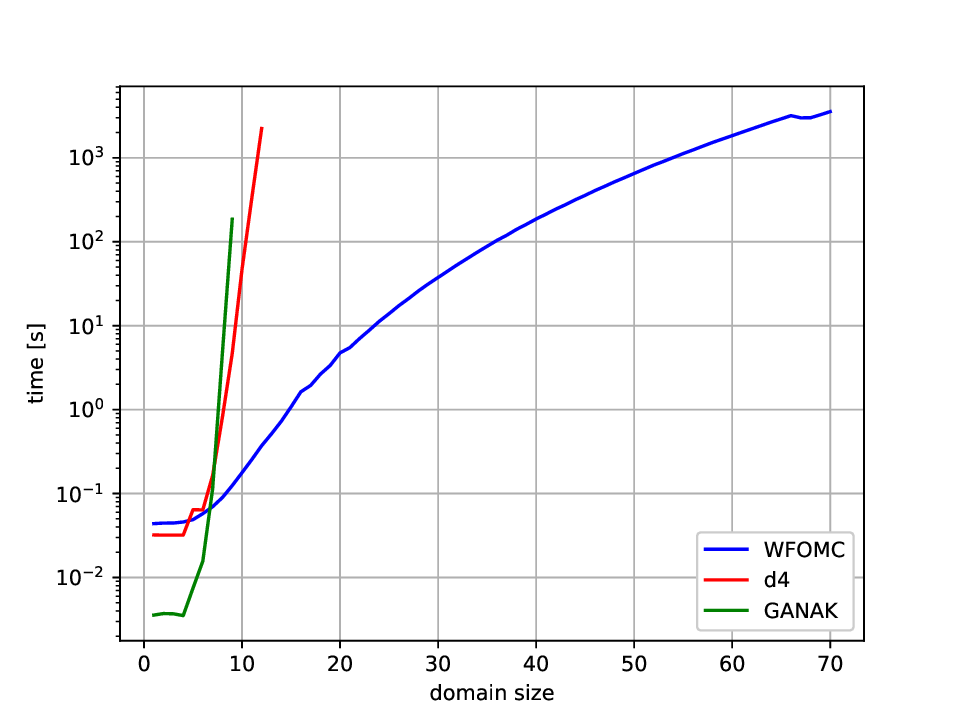}
    \caption{Runtime for $\Phi_2$}
  \end{subfigure}
  \begin{subfigure}[t]{.32\textwidth}
    \centering
    \includegraphics[width=\linewidth]{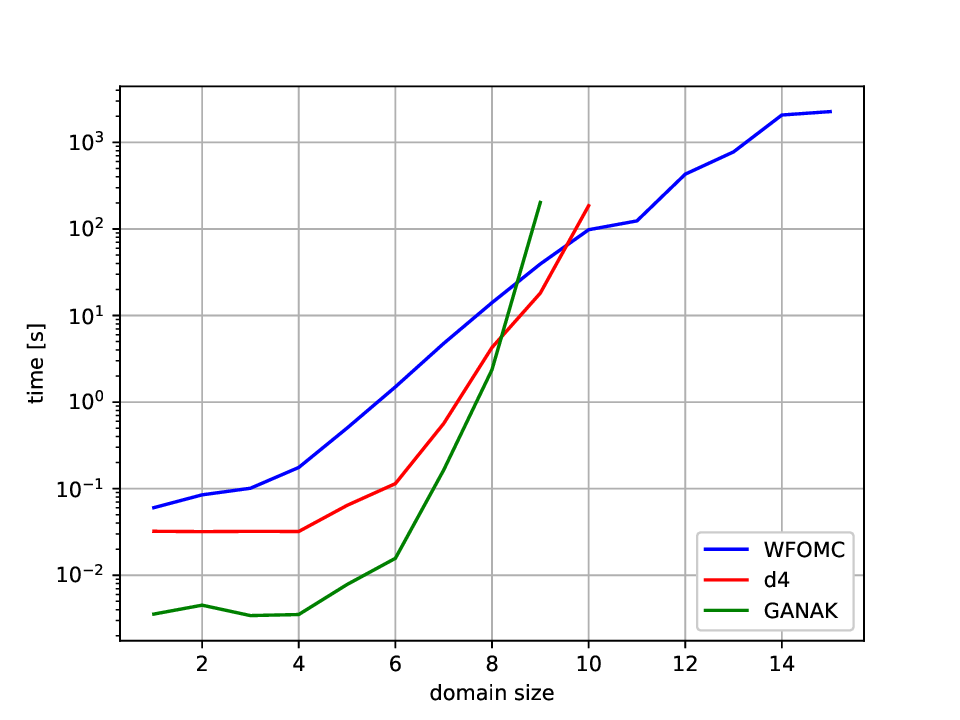}
    \caption{Runtime for $\Phi_4$}
  \end{subfigure}
  \begin{subfigure}[t]{.32\textwidth}
    \centering
    \includegraphics[width=\linewidth]{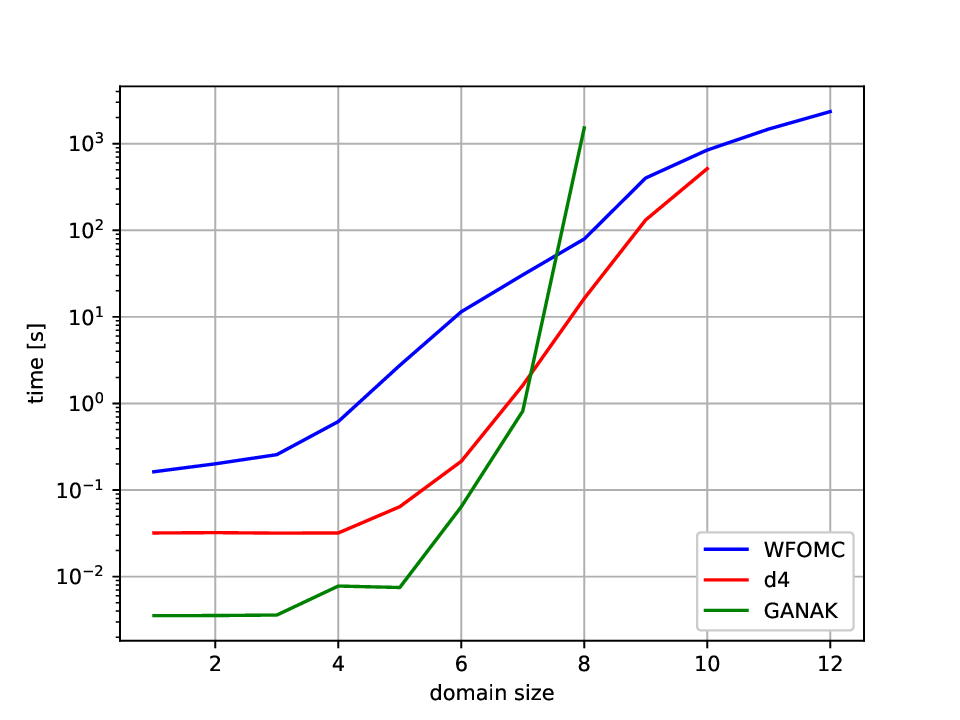}
    \caption{Runtime for $\Phi_5$}
  \end{subfigure}
  \begin{subfigure}[t]{.32\textwidth}
    \centering
    \includegraphics[width=\linewidth]{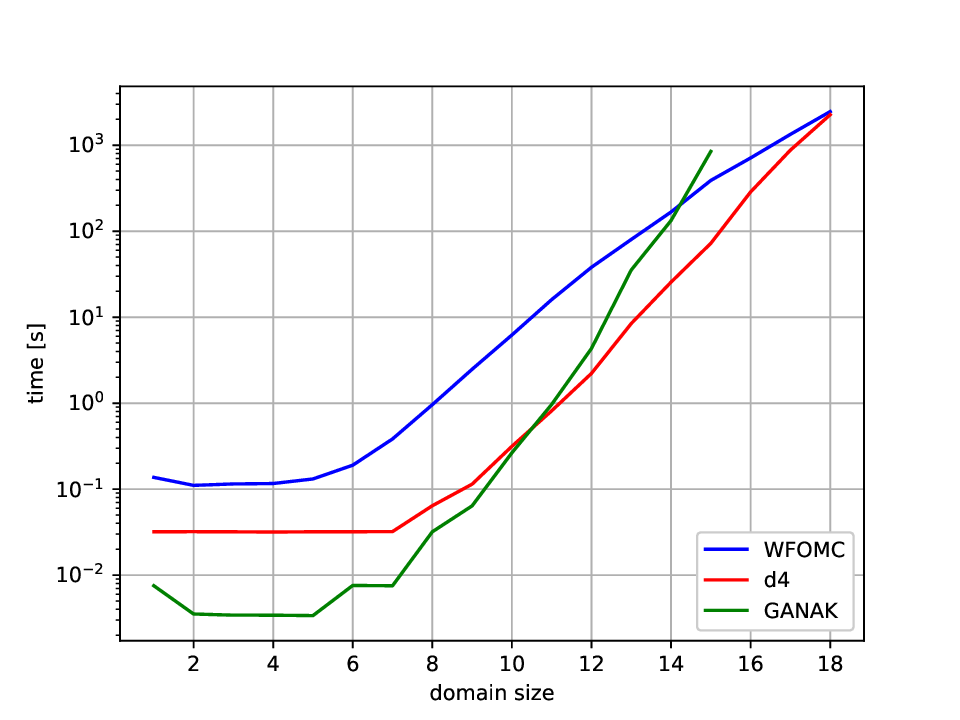}
    \caption{Runtime for $\Phi_{train}$}
  \end{subfigure}
  \caption{Execution times}
  \label{fig:execution}
\end{figure}

\subsection{Proofs of Matchings to OEIS}

\begin{lemma}
$\fomc(\Phi_1, n)$ equals to the $n$-th term of A000670, which is the number of ways to partition $n$ elements to disjoint subsets and arrange them into a sequence.
\end{lemma}

\begin{proof}
The sequence can be expanded to a permutation in the way that for each subset, we write down the numbers in descending order. For $\Phi_1$ when fixing the linear order $L$, the successor relation $S$ represents a permutation of domain elements. The binary predicate $B(x,y)$ indicates whether two adjacent elements $x$ and $y$ belong to the same subset. According to the expansion, if $L(x,y)$, they must belong to different subsets. But if $L(y,x)$, both cases are allowed. Therefore, there is a bijective mapping from the models of $\Phi_2$ to arrangements of partitions.
\end{proof}

\begin{lemma}
$\fomc(\Phi_2, n)$ equals to the $n$-th term of A000629, which is the number of ways to partition $n+1$ elements to disjoint subsets and arrange them into a necklace.
\end{lemma}

\begin{proof}
One can always break the necklace to a sequence by picking the subset containing element $n+1$ and concatenate the subsets in clockwise. In $\Phi_2$ we color each element by two colors: $U(x)$ being true and $U(x)$ being false. One can construct the bijective mapping from the models of $\Phi_2$ to arrangements of partitions by requiring that two adjacent elements are in the same subset if and only if they have the same color. Since $n+1$ is always the first element in the permutation, we fix its color as $U(x)=\top$, omit it and count the colorings in the remaining $n$ elements.
\end{proof} 

\section{Discussions}


The first interesting question from this work is whether there is a two-variable fragment with axioms on multiple relations whose FOMC is \class{\#P_1}-hard. Since FOMC of a first order logic sentence itself is a \class{\#P_1} problem, proving such hardness helps to establish more complete problems in the class \class{\#P_1}. The reductions used in our proof of hardness might not be applicable for FOMC as we require the usage of counting quantifiers and cardinality constraints which forces the reduction to WFOMC.

This work is an exploration of the boundary of WFOMC for \FOtwo{} sentences with multiple axioms. As only limited axioms are discussed in this paper, the complexity boundary for other interesting axioms remains unknown, especially the symmetric axioms such as the equivalence axiom, the connectedness axiom and the tree axiom. Ignoring our positive result, a natural attempt to prove the hardness of two such axioms is to force the relations to be successors of some linear orders and try to encode a grid by successors. Unfortunately, \Cref{thm:2succ} negates this idea and therefore encodings of other structures by these axioms should be established.

It can also be seen that our results coincide with the results in the decidability of finite satisfiability. For example, the work \cite{fs-FO2+2successor} proved that \FOtwo{} with two successor relations are decidable while \FOtwo{} with two linear order relations and their successor relations are undecidable. In our paper, it is also shown that \FOtwo{} with only two successor axioms are domain-liftable for WFOMC while strengthening both successor axioms to linear order axioms leads to \class{\#P_1}-hardness. Though the techniques for hardness proofs in the two problems are not identical, essentially both are trying to encode a grid and further a hard problem on the grid (e.g., the tiling problem) by the first-order logic fragments. An open question arises naturally whether there is a unified way to transform the undecidability of finite satisfiability problem to \class{\#P_1}-hardness of WFOMC, or conversely, the domain-liftablity to decidability.

\section*{Acknowledgments} V\'{a}clav K\r{u}la's work was supported by the CELSA project ``Towards Scalable Algorithms for Neuro-Symbolic AI''. Ond\v{r}ej Ku\v{z}elka's work was supported by the Czech Science Foundation project ``The Automatic Combinatorialist'' (24-11820S). 



\bibliographystyle{alpha}
\bibliography{ref}

\end{document}